%% file: main.tex
\documentclass{amsart}

\usepackage[latin1]{inputenc}
\usepackage{amsmath}
\usepackage{amsthm}
\usepackage{amssymb}
\usepackage{latexsym}
\usepackage{xspace}
\usepackage{tikz}

\input{./definitions}

\title
{On Guarded Transformation in the Modal Mu-Calculus}

\begin{document} 

\title{On Guarded Transformation in the Modal Mu-Calculus}

\author{Florian Bruse}\address{Florian Bruse \\ School of Elect.\ Eng.\ and Comp.\ Sc., University of Kassel \\ Germany}\email{florian.bruse@uni-kassel.de}
\author{Oliver Friedmann}\address{Oliver Friedmann \\ Dept.\ of Comp.\ Sc., University of Munich \\ Germany}\email{oliver.friedmann@ifi.lmu.de}
\author{Martin Lange}\address{Martin Lange \\ School of Elect.\ Eng.\ and Comp.\ Sc., University of Kassel \\ Germany}\email{martin.lange@uni-kassel.de}

\begin{abstract}
\input{./abstract}
\end{abstract}

\thanks{This work was supported by the European Research Council under the European Community's Seventh Framework Programme [ERC grant agreement no 259267]}
\thanks{Preprint submitted to: \emph{Logic Journal of the IGPL}}

\maketitle

\input{./intro}
\input{./prel}
\input{./oldgt}
\input{./lower}
\input{./eval}

\bibliographystyle{plain}
\bibliography{./literature}

\end{document}

%% file: definitions.tex
\newcommand{\ignore}[1]{}

\newcommand{\Nat}{\mathbb{N}}

\newcommand{\Transsys}{\ensuremath{\mathcal{T}}}
\newcommand{\Prop}{\ensuremath{\mathcal{P}}}
\newcommand{\Act}{\ensuremath{\Sigma}}
\newcommand{\Var}{\ensuremath{\mathcal{V}}}
\newcommand{\true}{\ensuremath{\mathtt{t\!t}}}
\newcommand{\false}{\ensuremath{\mathtt{f\!f}}}
\newcommand{\Mubox}[1]{\ensuremath{[ #1 ]}}
\newcommand{\Mudiam}[1]{\ensuremath{\langle #1 \rangle}}

\newcommand{\mucalc}{\ensuremath{\mathcal{L}_\mu}\xspace}
\newcommand{\boolcalc}{\ensuremath{\mathcal{B}_\mu}\xspace}
\newcommand{\hes}{HES\xspace}

\newcommand{\sub}[1]{\ensuremath{\operatorname{Sub}(#1)}}
\newcommand{\cl}[1]{\ensuremath{\operatorname{Cl}(#1)}}

\newcommand{\md}[1]{\ensuremath{\operatorname{md}(#1)}}

\newcommand{\sem}[3]{\ensuremath{{[\![#1]\!]}_{#2}^{#3}}}

\def\newarrow#1{\mathop{{\hbox{\setbox0=\hbox{$\scriptstyle{#1\quad}$}{$%
\mathrel{\mathop{\setbox1=\hbox to
\wd0{\rightarrowfill}\ht1=3pt\dp1=-2pt\box1}\limits^{#1}}%
$}}}}}
\newcommand{\Transition}[3]{\ensuremath{#1 \newarrow{#2} #3}}

\renewcommand{\epsilon}{\varepsilon}

\newcounter{linenumber}

\theoremstyle{plain}
\newtheorem{theorem}{Theorem}[section]
\newtheorem{lemma}[theorem]{Lemma}

\newtheorem{proposition}[theorem]{Proposition}

\theoremstyle{definition}

\newtheorem{example}[theorem]{Example}

\theoremstyle{plain}


%% file: abstract.tex
Guarded normal form requires occurrences of fixpoint variables in a $\mu$-calculus-formula to occur under the scope of a modal operator. The literature contains guarded transformations that effectively bring a $\mu$-calculus-formula into guarded normal form. We show that the known guarded transformations can cause an exponential blowup in formula size, contrary to existing claims of polynomial behaviour. We also show that any polynomial guarded transformation for $\mu$-calculus-formulas in the more relaxed vectorial form gives rise to a polynomial solution algorithm for parity games, the existence of which is an open problem. We also investigate transformations between the $\mu$-calculus, vectorial form and hierarchical equation systems, which are an alternative syntax for alternating parity tree automata. 

%% file: intro.tex
\section{Introduction}
The modal $\mu$-calculus \mucalc, as introduced by Kozen \cite{Kozen83}, is a fundamental modal fixpoint logic which subsumes many other
temporal \cite{Emerson90a,Dam:1994:CEF} and dynamic logics \cite{Kozen83,EL86}. The paper at hand is concerned with \emph{guarded normal
  form} for the $\mu$-calculus, or simply \emph{guarded form}. A formula is guarded if every occurrence of any fixpoint variable is under
the scope of a modal operator inside its defining fixpoint formula. For instance, the formula $\nu Y.\Diamond\mu X.p \vee \Diamond X$ is
guarded, whereas $\nu Y.\mu X.(p \wedge Y) \vee \Diamond X$ is not: it is possible to pass through the syntax tree of the formula from the
quantification for the variable $Y$ down to on occurrence of that variable without traversing through a modal operator $\Diamond$ or $\Box$.
This is not possible for the first formula given here.

Intuitively, guarded form ensures that in the evaluation of a formula in a transition system by fixpoint iteration, one proceeds along at least one transition between two iterations of the same variable.  Guarded form is also very useful in procedures that check for satisfiability or validity of a set of formulas and handle fixpoint formulas by unfolding: guardedness synchronises the unfolding of all formulas in a set. Many constructions require formulas to be explicitly normalised in guarded form or assume that w.l.o.g., formulas can be brought into guarded form with polynomial overhead \cite{CONCUR::Kaivola1995,CONCUR::JaninW1996,KVW00,Walukiewicz/00,conf/tacas/Mateescu02,conf/tableaux/Jungteerapanich09}. Others can cope with unguarded formulas, but then their constructions require the solving of non-trivial decision problems \cite{focs91*368}. Only few deal explicitly with unguarded formulas \cite{conf/tableaux/FriedmannL11}, but they require special tricks in order to handle unguardedness.

It has been known for quite a while that every \mucalc-formula can effectively be transformed into an equivalent guarded formula. The first guarded transformation routine---described by Banieqbal and Barringer, as well as Walukiewicz---explicitly rewrites Boolean subformulas into disjunctive or conjunctive normal form \cite{BB:temlfp,Walukiewicz/00}. Clearly, this increases the sizes of formulas exponentially in the worst case. Such a blowup may not be considered harmful for results concerning the expressive power of \mucalc, but it clearly makes a difference for complexity-theoretic results. Kupferman et al.\ \cite[Thm.~2.1]{KVW00} notice that the transformation into Boolean normal form is unnecessary and present an optimised variant of this guarded transformation procedure. They claim that it only involves
a linear blowup, but this is not true. The problem lies in the very last statement of their proof: ``\ldots~by definition $\varphi'(\lambda y.\varphi'(y)) \in \cl{\lambda y.\varphi'(y)}$~\ldots'' While this is true, it is not true that if $\lambda y. \varphi'(y)$ is a subformula of $\varphi$---and hence in the Fischer-Ladner closure of $\varphi$---then also $\varphi[\varphi'(\lambda y. \varphi'(y))/\lambda y. \varphi'(y)]$ is in the closure of $\varphi$. Unfolding from the inside out---a principle that forms the core of the guarded transformation---produces exactly this kind of situation. Repeated application of this principle will, in the worst case, result in an exponential growth in the number of distinct subformulas.
 
Later, another guarded transformation procedure was given by Mateescu \cite
{conf/tacas/Mateescu02}. It turns out that
it is practically the same algorithm as that given by Kupferman et al.\ earlier. However, Mateescu estimates it to 
create an exponential blowup in formula size when used na\"{\i}vely. On the other hand, he claims that ``\ldots~each
fixpoint subformula \ldots will be duplicated only once \ldots and reused \ldots leading to \ldots $|t(\varphi)| \le |\varphi|^2$.'' 
Again, this is a false observation because the replacement of certain variables by constants can duplicate subformulas.
Thus, the formula sharing trick that Mateescu proposes as a solution, which is just a way of saying that the size measured
in terms of number of subformulas shall only be quadratic, does not work and the blow-up is not just quadratic. 

Neumann and Seidl study guarded transformation in the more general context of hierarchical equation systems \cite{conf/csl/SeidlN99} with
monotone operators. The modal operators $\Box$ and $\Diamond$ are monotone, so hierarchical equation systems are a generalisation of the
\emph{vectorial form} that is sometimes used to put multiple nestings of least or greatest fixpoints of the same kind into one block
\cite{an01}. Hierarchical equation systems can also be seen as an alternative syntax to present alternating parity automata. The semantics
of \mucalc with vectorial form is simply given via simultaneous fixpoint definitions rather than parametric ones. The Beki\'c Lemma shows
that this does not gain additional expressive power \cite{Bekic84} but it may gain exponential succinctness because the only known
transformations of \mucalc with vectorial form into \mucalc without incur an exponential blow-up in formula size. The same holds for
hierarchical equation systems.

Neumann and Seidl even give a guarded transformation algorithm for equation systems that result from an \mucalc-formula and claim that it is
polynomial. This claim seems to be correct. However, the resulting equation system does not have the nice structure that
equation systems corresponding to \mucalc-formulas have. Thus, their guarded transformation for \mucalc-formulas is polynomial, yet it does
not produce equivalent guarded \mucalc-formulas but only equivalent guarded equation systems. Translating these back into a guarded
\mucalc-formula incurs an exponential blowup, given current knowledge.

Here we study the problem of guarded transformation for the modal $\mu$-calculus with the aim of correcting false claims found in the
literature and providing bounds on the complexities of such transformations.  The structure of the paper is as follows: in
Section~\ref{sec:prel} we introduce \mucalc and its syntactic extensions vectorial form and \hes as well as several notions of guardedness.
Then we briefly describe and analyse the guarded transformation procedures by Kupferman et al.\ and Mateescu, respectively that of Neumann
and Seidl in Section~\ref{sec:oldgt} and show that it can produce formulas of at least exponential size (measured as the number of different
subformulas). We also give additional upper bounds on the elimination of $\epsilon$-transitions in alternating automata and on translations
from a syntactic extension to flat \mucalc. In Section~\ref{sec:lower} we show that guarded transformation for \mucalc-formulas in vectorial
form and for \hes is hard, namely not easier than solving parity games. As mentioned above, we also show that unfolding
vectorial form or an \hes into a non-vectorial formula has the same lower complexity bound. This means that any polynomial algorithm for one
of these problems would yield a polynomial algorithm for solving parity games, and this would settle a major and long-standing open problem.
Finally, in Section~\ref{sec:eval} we discuss the consequences of this work for previously acclaimed results about \mucalc that can be found
in the literature, we discuss the relation between alternating automata, vectorial form, and \mucalc, and we sketch some open questions with
possible routes of attack.


%% file: prel.tex
\section{The Modal $\mu$-Calculus}
\label{sec:prel}

\subsection{Syntax and Semantics}
A \emph{labeled transition system} (LTS) over a set of \emph{action names} $\Act$ and a set of \emph{atomic propositions}
$\Prop$ is a tuple $\Transsys = (S,\Transition{}{}{},\ell)$ where $S$ is a set of \emph{states}, 
$\Transition{}{}{} \subseteq S \times \Act \times S$ defines a set of \emph{transitions} between states, labeled with action names, and $\ell: S \to 2^{\Prop}$ labels each state with the set of atomic propositions that are true in this state.

Let $\Act$ and $\Prop$ be as above and let $\Var$ be a set of variables. Formulas of the modal $\mu$-calculus 
\mucalc in positive normal form are those that can be derived from $\varphi$ in 
\[
\varphi \enspace ::= \enspace q \mid \overline{q} \mid X \mid \varphi \vee \varphi \mid \varphi \wedge \varphi
  \mid \Mudiam{a}\varphi \mid \Mubox{a}\varphi \mid \mu X.\varphi \mid \nu X.\varphi,
\]
where $X \in \Var$, $q \in \Prop$, and $a \in \Act$.

The operators $\mu$ and $\nu$ act as \emph{binders} for the variables in a formula.  A \emph{free occurrence} of a variable
$X$ is therefore one that does not occur under the scope of such a binder. A \emph{closed} formula is one that does not
have any free variables. We write $\sigma$ for either $\mu$ or $\nu$.

Let $\sub{\varphi}$ denote the set of \emph{subformulas} of $\varphi$. Define the \emph{size} of a formula $\varphi$ as
the number of its distinct subformulas, i.e.\ $|\varphi| := |\sub{\varphi}|$. We assume all $\mucalc$-formulas to be \emph{well-named} in the sense that each variable is bound at most once.  Hence, for every \mucalc-formula there is a partial function $\operatorname{fp}_{\varphi}: \Var \to \sub{\varphi}$ which maps a variable $X$ that is bound in $\varphi$ by some operator $\sigma X.\psi$ to its \emph{defining fixpoint formula} $\psi$. 

As usual, we use the abbreviations $\true = q \vee \overline{q}$ and $\false = q \wedge
\overline{q}$ for an arbitrary $q$. Given a fixpoint binder $\sigma$, we write $\hat \sigma = \false$ if $\sigma = \mu$ and $\hat \sigma = \true$ if $\sigma = \nu$.

We write $\varphi[\psi/X]$ to denote 
the formula that results from $\varphi$ by replacing every free occurrence of the variable $X$ in it with the 
formula $\psi$. 

 The \emph{modal depth} $\operatorname{md}$ is the maximal nesting depth of modal operators in a formula, formally defined as follows.
\begin{align*}
\md{q} \enspace = \enspace \md{\overline{q}} \enspace = \enspace \md{X} \enspace &:= \enspace 0 \\
\md{\varphi \vee \psi} \enspace = \enspace \md{\varphi \wedge \psi} \enspace &:= \enspace \max \{\md{\varphi}, \md{\psi} \} \\
\md{\Mudiam{a}\varphi} \enspace = \enspace \md{\Mubox{a}\varphi} \enspace &:= \enspace 1 + \md{\varphi} \\
\md{\mu X.\varphi} \enspace = \enspace \md{\nu X.\varphi} \enspace &:= \enspace \md{\varphi}
\end{align*}

An important fragment of \mucalc that we consider later is the \emph{propositional $\mu$-calculus} \boolcalc. It consists of all $\varphi \in \mucalc$ such that $\md{\varphi} = 0$. Hence, \boolcalc-formulas do not contain any subformulas of the form $\Mudiam{a}\psi$ or $\Mubox{a}\psi$. A formula is called \emph{purely propositional} if it belongs to \boolcalc and does not contain any fixpoint operators.

Formulas of \mucalc are interpreted in states of an LTS $\Transsys = (S,\Transition{}{}{},\ell)$. 
Let $\rho\colon\Var \to 2^{\mathcal{S}}$ be an environment used to interpret 
free variables. We write $\rho[X \mapsto T]$ to denote the environment which maps $X$ to $T$ and behaves like $\rho$ 
on all other arguments. The semantics of \mucalc is given as a function $\sem{\cdot}{}{}$ mapping a formula to the set of states where it holds w.r.t.\ the environment.
\begin{align*}
\sem{q}{\rho}{\Transsys}                   &= \{ s \in S \mid q \in \ell(s) \} \\
\sem{\overline{q}}{\rho}{\Transsys}        &= \{ s \in S \mid q \not\in \ell(s) \} \\
\sem{X}{\rho}{\Transsys}                   &= \rho(X) \\
\sem{\varphi \vee \psi}{\rho}{\Transsys}   &= \sem{\varphi}{\rho}{\Transsys} \cup \sem{\psi}{\rho}{\Transsys} \\
\sem{\varphi \wedge \psi}{\rho}{\Transsys} &= \sem{\varphi}{\rho}{\Transsys} \cap \sem{\psi}{\rho}{\Transsys} \\
\sem{\Mudiam{a}\varphi}{\rho}{\Transsys}   &= 
  \{ s \in S \mid \exists t \in \sem{\varphi}{\rho}{\Transsys} \mbox{ with } \Transition{s}{a}{t} \} \\
\sem{\Mubox{a}\varphi}{\rho}{\Transsys}    &=
  \{ s \in S \mid \forall t \in S: \mbox{ if } \Transition{s}{a}{t} \mbox{ then } t \in \sem{\varphi}{\rho}{\Transsys} \} \\
\sem{\mu X.\varphi}{\rho}{\Transsys}       &= \bigcap \{ T \subseteq S \mid \sem{\varphi}{\rho[X \mapsto T]}{\Transsys} \subseteq T \} \\
\sem{\nu X.\varphi}{\rho}{\Transsys}       &= \bigcup \{ T \subseteq S \mid T \subseteq \sem{\varphi}{\rho[X \mapsto T]}{\Transsys} \} 
\end{align*}
Two formulas $\varphi$ and $\psi$ are equivalent, written $\varphi \equiv \psi$, iff for all LTS $\Transsys$ and all environments 
$\rho$ we have $\sem{\varphi}{\rho}{\Transsys} = \sem{\psi}{\rho}{\Transsys}$. We may also write $\Transsys, s \models_\rho \varphi$ 
instead of $s \in \sem{\varphi}{\rho}{\Transsys}$. 

\subsection{Syntactic Extensions} 
Sometimes it is convenient to relax the restrictions on variable dependency. \emph{Vectorial form} allows one to do so. Let $X_1,\ldots,X_m$ be
variables and let $\psi_1,\ldots,\psi_m$ be formulas, possibly with free occurrences of the $X_i$. For both $\sigma \in \{\mu,\nu\}$ and any
$1 \leq j \leq m$,
\[
\Phi = \sigma X_j. \left\{\begin{aligned} X_1&.&&\psi_1 \\ &\vdots \\ X_m&.&&\psi_m \end{aligned}\right\}
\]
is a formula in \emph{$m$-vectorial form}, and the $X_i$ are considered bound in $\Phi$.  We say that a formula is in \emph{vectorial form} if it is in $m$-vectorial form for some $m$. Hence, formulas in $1$-vectorial form are ordinary \mucalc-formulas as introduced above. The curly brackets are used to indicate that the $m$ defining fixpoint equations are to be seen as a set; there is no implicit order among
them with the exception of a variable $X_j$ marked as \emph{entry variable}. By convention, the entry variable is the variable that occurs first if it is not explicitly given. We consider multiple instances of the same vectorial form, but with different entry variable, to be the essentially the same subformula in the sense that $k$  occurrences of the same block of size $k'$ with different entry variables contribute to the size of the formula with $k + k'$ instead of $kk'$.

The semantics of a vectorial formula is defined as follows: given a transition system $\Transsys = (S, \Transition{}{}{}, \ell),$ and an environment $\rho$, a vector of fixpoint formulas
\[ \Psi =  \left\{\begin{aligned} X_1&.&&\psi_1 \\ &\vdots \\ X_m&.&&\psi_m \end{aligned}\right\}\]
defines a monotone operator $S^m \to S^m$ via $(S_i)_{i \leq m} \mapsto \sem{\varphi_i}{\rho[X_1\mapsto S_1,\dotsc X_m \mapsto S_m]}{\Transsys}$. Let $T = (T_1,\dotsc, T_m) \subseteq S^m$ be the least fixpoint of this operator, and let $X_k$ be the entry variable. Then
\[
\sem{\mu X_k. \Psi}{\rho}{\Transsys} = T_k
\]
and accordingly for greatest fixpoints. 

\begin{example}
\label{ex:vectrans}
Consider the following 3-vectorial formula where $\Box\psi$ is used to abbreviate $\bigwedge_{a \in \Sigma} \Mubox{a}\psi$.
\[
\mu X.\left\{\begin{aligned}
X&.&& \Box\false \vee \Mudiam{a}Y \vee Z \\
Y&.&& \Mudiam{b}(Y \vee X) \\
Z&.&& \Mudiam{a}X \vee \Mudiam{c}Z
\end{aligned}\right\}
\]
It expresses ``there is a maximal path labelled with a word from $(ab^+ + c^*a)^*$''. It abbreviates the following
formula in non-vectorial form. 
\[
\mu X.\Box\false \vee \Mudiam{a}(\mu Y.\Mudiam{b}(Y \vee X)) \vee \mu Z.\Mudiam{a}X \vee \Mudiam{c}Z
\]
\end{example}

\emph{Hierarchical equation systems}, or \hes, generalise \mucalc-formulas by lifting the restriction of tree-like variable dependencies (see below). An equation has the form $Z = e_Z$, where $Z \in \Var$ and $e_Z$ can be generated from $e$ in
\[
e \enspace ::= \enspace q \mid \overline{q} \mid X \mid e \vee e \mid e \wedge e
  \mid \Mudiam{a} e \mid \Mubox{a} e 
\]
with $X \in \Var$, $q \in \Prop$, and $a \in \Act$ as above. An \hes is a finite set $\mathcal{S}$ of equations with disjoint left sides, an ordered partition $\{(S_1, \sigma_1), \dotsc, (S_k, \sigma_k)\}$ of the equations with fixpoint qualifications and a designated entry variable from the first partition class. Here, $\sigma_i \in \{\mu, \nu\}$ as per usual. A partition class $S_i$ is also called a block. The size of an \hes $\mathcal{S}$ with left hand sides in $Z$ is defined as $\Sigma_{z \in Z} |e_z|$. A variable is \emph{free} in $S_i$ if it is not the left hand side of an equation in some $S_j$ for $j \geq i$. An \hes is called closed if no variable is free in $S_1$.
An equation system is called  \emph{boolean} if it is built without use of $\Mudiam{a}$ and $\Mubox{a}$ for any $a \in \Act$.

Following Neumann and Seidl \cite{conf/csl/SeidlN99}, we define the semantics of \hes on the powerset lattice of a given transition system $\Transsys = (S, \Transition{}{}{}, \ell)$. For a system $\mathcal{S}$ with partition $\{(S_1,\sigma_1),\dotsc,(S_k, \sigma_k)\}$, define $Z_i$ as the set of left sides of equations in $S_i$. For two assignments $\rho_1$ and $\rho_2$ with disjoint domains and equal range, define $\rho_1 + \rho_2$ as the assignment that behaves like $\rho_1$ on the domain of $\rho_1$ and behaves like $\rho_2$ on the domain of $\rho_2$. The semantics of $\mathcal{S}$ is defined inductively over the blocks. Given an assignment $\rho\colon\operatorname{free}(S_i)\to 2^{S}$, a block defines a function $f\colon (Z_i \to 2^S) \to (Z_i \to 2^S)$ via
\[
(\tau,Z) \mapsto \sem{e_z}{\rho+\tau+\sem{S_{i+1}}{\rho+\tau}{\Transsys}}{\Transsys}
\]
where $\sem{S_{i+1}}{\rho + \tau}{\Transsys}$ is the empty assignment if $i = k$. Let $\hat{f}$ denote the least fixpoint of $f$ with respect to $\tau$ if $\sigma_i = \mu$, respectively the greatest such fixpoint if $\sigma_i = \nu$. 
Define  $\sem{S_i}{\rho}{\Transsys}(Z)$ as

\[
\sem{S_i}{\rho}{\Transsys}(Z) = \left\{ 
\begin{aligned}
&\hat{f}(Z)                                     &&\text{    if }Z\in Z_i \\
&\sem{S_{i+1}}{\rho + \hat{f}}{\Transsys}(Z)    &&\text{    if }Z \in Z_j, j > i
\end{aligned}
\right.
\]
Finally, given an assignment $\rho\colon\operatorname{free}(S_1)\to 2^S$, set $\sem{\mathcal{S}}{}{\Transsys} = \sem{S_1}{\rho}{\Transsys}$ and for closed $\mathcal{S}$ write $\Transsys,s \models \mathcal{S}$ if $s \in  \sem{\mathcal{S}}{}{\Transsys}(Z)$, and the entry variable $Z$ is clear from the context or has been explicitly designated. Two \hes $\mathcal{S}$ and $\mathcal{S}'$ are equivalent if they agree on their free variables and their outermost variable and, for every $\Transsys = (S, \Transition{}{}{},\ell)$ and every $\rho\colon\operatorname{free}(\mathcal{S})\to 2^S$, we have $\sem{\mathcal{S}}{\rho}{\Transsys} = \sem{\mathcal{S}'}{\rho}{\Transsys}$.
   
It is not hard to see that we can always combine adjacent blocks with the same fixpoint qualification, hence we can consider the fixpoint qualifiers of the blocks to be strictly alternating.

Interestingly, \hes and alternating parity tree automata in their symmetric form are different forms of notation for the same thing: states of the automaton correspond to the variables in the equation systen, the block partition corresponds to the priority function over the states of an automaton, and the equations make up the transition function.

\subsection{Equational Form, Variable Order, and Variable Dependencies}

\mucalc-formulas and vectorial \mucalc-formulas can be brought into an alternative syntax that resembles \hes. Consider an \mucalc-formula $\varphi$ with fixpoint variables in $X_1,\dotsc,X_n$. Without loss of generality, $\varphi$ is of the form $\sigma X_1.\psi_1$.\footnote{Otherwise, introduce a vacuous outermost fixpoint quantifier. This can be done without altering alternation depth.} Proceeding from an innermost fixpoint formula, convert a subformula of the form $\sigma_i X_i. \psi_i$ to a single-equation block $(\{X_i = e_{X_i}\},\sigma_i)$ by a translation that abstracts away the difference between fixpoint quantification of the form $\sigma X. \psi$ and occurrences of fixpoint variables. This translation $e$ is defined via
\[
\begin{aligned}
e(q)                    &=  q  &
e(\Mudiam{a} \psi)      &=  \Mudiam{a} e(\psi) \\
e(\overline{q})         &=  \overline{q} &
e(\Mubox{a} \psi)       &=  \Mubox{a} e(\psi) \\
e(\psi_1 \wedge \psi_2) &=  e(\psi_1) \wedge e(\psi_2) \qquad &
e(\sigma' X'.\psi)      &=  X' \\
e(\psi_1 \vee \psi_2)   &=  e(\psi_1) \vee e(\psi_2) &
e(X)                    &=  X.
\end{aligned}
\]
This process yields a set of equation blocks $\{(\{X_1 = e_{X_1}\},\sigma_1),\dotsc,(\{X_n = e_{X_n}\},\sigma_n)\}$, partially ordered by the subformula relationship of the original formulas. We call the strict, transitive version of this order the \emph{priority order}. The equational syntax generalizes accordingly to formulas in vectorial form: fixpoint variables in the same vector share a block as in the case of \hes. As in the case of plain \mucalc-formulas, the subformula relationship induces a strict, transitive partial order. In this case, two variables from the same block are incomparable. Any linearization of this partial order produces a hierarchical equation system. 

Consider an \hes with fixpoint variables $\mathcal{X} = X_n,\dotsc,X_0$. Construct a \emph{variable dependency graph} with node set $\mathcal{X}$ as follows: there is an edge from variable $X_i$ to variable $X_j$ if and only if $X_j$ appears in $e_{X_i}$. 

Note that the variable dependency graph for an \hes resulting from an \mucalc-formula is always a tree with back edges compatible with the
priority order. This means that any edge in the dependency graph either goes from a node to an immediate successor in the priority order, or it is a loop, or it goes to a predecessor in the priority order. However, edges never go from a node to an indirect successor, e.g. to the son of a son. In this sense, priority order and variable dependency almost coincide and can be deduced from the subformula nesting.

The variable dependency graph of a vectorial formula or an \hes is generally not a tree with back edges. However, formulas in vectorial form can be characterized in the following way: a formula is in vectorial form if and only if the set of equations can be partitioned into blocks and the blocks can be partially ordered respecting the priority order such that variable dependency edges only stay in a block, go from one block to an immediate successor block, or go to a---not necessarily immediate---predecessor block. In this sense, the block structure forms a tree with back edges. Moreover, all equations in a block are of the same type, i.e., qualified only by $\mu$ or only by $\nu$.

\subsection{Guardedness and Weak Guardedness}

We extend the variable dependency graph by annotating edges with the information whether or not a modal operator has been passed between the
two variables. More precisely, we call an edge from variable $X$ to variable $Y$ guarded if in the equation $X = e_X$ variable $Y$ occurs
under the scope of a modal operator. Note that there can be both a guarded and an unguarded edge from a variable to another. Call the
resulting graph the \emph{guardedness graph}.

An occurence of a variable in an equation is called \emph{unguarded} if it is part of an unguarded cycle in the guardedness graph, and it is called \emph{guarded} if it is not unguarded. An equation system is called \emph{guarded} if there are no unguarded occurrences of variables in its guardedness graph, i.e., there are no unguarded cycles. An equation system is called \emph{downwards guarded} if unguarded edges only occur strictly upwards in the partial order induced by variable priority. An equation system is called \emph{$\epsilon$-free} if there are no unguarded edges in its guardedness graph. Clearly, $\epsilon$-freeness implies downwards guardedness, which in turn implies guardedness.

\begin{example}
The formula
\[
\label{eq:nonvecform}
\mu X.\Box\false \vee \Mudiam{a}(\mu Y.\Mudiam{b}(Y \vee X)) \vee \mu Z.\Mudiam{a}X \vee \Mudiam{c}Z
\]
from Example~\ref{ex:vectrans} expresses ``there is a maximal path labeled with a word from $(ab^+ + c^*a)^*$'' and is guarded. However, consider the following 
formula which expresses the slightly different property ``there is a maximal path labeled with a word from $(ab^+ + c^*)^*$''.
\[
\mu X.\Box\false \vee \Mudiam{a}(\mu Y.\Mudiam{b}(Y \vee X)) \vee \mu Z.X \vee \Mudiam{c}Z
\]
It is not guarded; in particular, there is an occurrence of $X$---the latter one---which is not guarded in its defining fixpoint formula,
which happens to be the entire formula in this case.
\end{example}

We say that an occurrence of a variable $X$ is \emph{weakly guarded} if it is guarded or if all unguarded cycles for this occurrence of this
variable have length at least $2$, i.e. $X$ does not appear unguarded in $e_X$. For plain \mucalc-formulas, this means that $X$ is either
guarded or it occurs under the scope of another fixpoint quantifier in its defining fixpoint subformula $\sigma X.\psi_X$. Note that weak
guardedness is indeed weaker than guardedness, hence, not being weakly guarded entails not being guarded.

Consider, for instance, the formula $\mu X.q \vee (\mu Y. (q \wedge X) \vee (\overline{q} \wedge Y) \vee \Mudiam{a}Y)$. Then $Y$ has both a guarded and an unguarded occurrence, whereas the only occurrence of $X$ is not guarded but it is weakly guarded.

A \emph{guarded transformation} for \mucalc or its syntactic extensions is a function 
such that $\tau(\varphi)$ is guarded and $\tau(\varphi) \equiv \varphi$ for every $\varphi \in \mucalc$. 


%% file: oldgt.tex
\section{Upper Bounds and Failure Results}
\label{sec:oldgt}

\subsection{Guarded Transformation Without Vectorial Form}

The guarded transformation procedures for non-vectorial formulas by Kupferman et al.\ and Mateescu rely on two principles.
The first principle is the well-known fixpoint unfolding.

\begin{proposition}
\label{prop:unfold}
For every $\sigma X.\varphi \in \mucalc$ we have $\sigma X.\varphi \equiv \varphi[\sigma X.\varphi / X]$. 
\end{proposition}

The second principle states how occurrences that are not weakly guarded can be eliminated. Remember that $\hat{\sigma}$ is
either $\true$ or $\false$ depending on $\sigma$ being $\nu$ or $\mu$.

\begin{proposition}[\cite{KVW00,conf/tacas/Mateescu02}]
\label{prop:replacenwg}
Let $\sigma X.\varphi \in \mucalc$ and let $\sigma X.\varphi'$ result from $\sigma X.\varphi$ by replacing with  $\hat{\sigma}$ every occurrence
of $X$ that is not weakly guarded. Then $\sigma X.\varphi \equiv \sigma X.\varphi'$.
\end{proposition}

These two principles can be combined to a simple guarded transformation procedure. Starting with the innermost
fixpoint bindings, one replaces all occurrences of the corresponding variables that are not weakly guarded by $\true$ or 
$\false$ using Proposition~\ref{prop:replacenwg}. Note that for the innermost fixpoint subformulas, the concepts of 
being weakly guarded and being guarded coincide. Thus, the innermost fixpoint subformulas are guarded after this step. 

For outer fixpoint formulas this only ensures that all remaining occurrences are weakly guarded. However, 
by the induction hypothesis, all inner ones are already guarded, and unfolding them using Prop.~\ref{prop:unfold} puts all
weakly guarded but unguarded occurrences of the outer variable under a $\Mudiam{a}$- or $\Mubox{a}$-modality.
Hence, only occurrences that are either guarded or not weakly guarded survive, and the latter can be eliminated
using Proposition~\ref{prop:replacenwg} again.

Let $\tau_0$ denote the guarded transformation which works as described above. Kupferman et al.\ claim that the worst-case
blowup in formula size produced by $\tau_0$ is linear, Mateescu claims that it is quadratic. We will show 
that it is indeed exponential. Consider the family of formulas
\[
\Phi_n \enspace := \enspace \mu X_1 \ldots \mu X_n. (X_1 \vee \dotsb \vee X_n) \vee  \Mudiam {a}(X_1 \vee \dotsb \vee X_n) .
\]
\begin{theorem}
We have $|\Phi_n| = 3n+1$ and $|\tau_0(\Phi_n)| = \Omega(2^n)$.
\end{theorem}
\begin{proof}
The first claim about the linear growth of $\Phi_n$ is easily verified. We prove that $\tau_0(\Phi_n)$ 
contains a subformula of modal depth at least $2^{n-1}$, which entails exponential size of $\tau_0(\Phi_n)$.

Let $\varphi =  (X_1 \vee \dotsb \vee X_n)$.
Mateescu's guarded transformation transforms  a (strict) subformula of the form $\sigma X. \psi$, into $f_X(t'(\psi))[\sigma X. f_X(t'(\psi))/X]$, where $t'$ is the guarded transformation for subformulas and $f_X$ replaces unguarded occurrences of $X$ by $\hat{\sigma}$. Moreover, $t'(\varphi \vee \Mudiam{a} \varphi) = \varphi \vee \Mudiam{a}\varphi$. For $2 \leq i \leq n$, define $\varphi_i = t'(\mu X_i \dotsb X_n. (\varphi \vee \Mudiam{a}\varphi)$. Then $\varphi_n = f_{X_n}(\varphi \vee \Mudiam{a} \varphi)[\mu X_n. f_{X_n}(\varphi \vee \Mudiam{a} \varphi)/X_n]$, and generally, $\varphi_i = f_{X_i}(\varphi_{i+1})[\mu X_i. f_{X_i}(\varphi_{i+1})/X_i]$. We show that $\varphi_i$ contains $\varphi$ at modal depth $2^{(n+1-i)}$. Clearly $\varphi_n$  contains $\varphi$ at modal depth $2 = 2^1$.  Since $\varphi_{i} = f_{X_i}(\varphi_{i+1})[\mu X_i. f_{X_i}(\varphi_{i+1})/X_i]$, we have that, if $\varphi_{i+1}$ contains $\varphi$ at modal depth $2^{n-i}$, then $\varphi_i$ contains $\varphi$ at double the modal depth, or $2\cdot 2^{n-i} = 2^{n+1-i}$. Hence, $t'(\mu X_2 \dotsb \mu X_n. (\varphi \vee \Mudiam{a}\varphi)) = \varphi_2$ contains a formula of modal depth $2^{n-1}$. Finally, Mateescu defines $\tau_0(\sigma X. \psi) = \sigma X. f_X(t'(\psi))$, whence $\tau_0(\Phi_n) = \mu X_1. f_X(\varphi_2)$, and the proof is finished.
\end{proof}
However, the blowup for this guarded transformation procedure is never worse than exponential.
\begin{theorem}
For all every \mucalc-formula $\varphi$, the size of $\tau_0(\varphi)$ is in $\mathcal{O}(2^{|\varphi|})$.
\end{theorem}
\begin{proof}
  Let $\varphi \in \mucalc.$ We analyze the size of $\tau_0(\varphi)$ from the inside out. Replacing unguarded occurrences of the innermost
  fixpoint quantifier with the default values will at most double the size of this formula. Unfolding it will double the size again. For a
  non-innermost, non-outermost quantifier, replacement of non-weakly guarded occurrences of variables will only affect subformulas created
  by unfolding. Since these are already accounted for in the blowup, changing some occurrences of variables to default values will not
  contribute to the blowup. Unfolding a non-innermost, non-outermost subformula will double its size at worst. The outermost formula is not
  unfolded. This makes at most one doubling of size per fixpoint quantifier, hence the resulting size is bounded from above by
  $2^{|\varphi|}$.
\end{proof}
Note that this guarded transformation procedure always produces downwards guarded formulas: occurrences of variables are always guarded, and fixpoint quantifiers are also always under the scope of a modal operator.

\subsection{Guarded Transformation With Vectorial Form and for \hes}

We do not study guarded transformation for vectorial formulas in particular, because by Lemma~\ref{hes-to-vec}, \hes can be converted into vectorial formulas with only polynomial blowup. This transformation keeps guardedness, but the resulting vectorial formula will not be downwards guarded and, hence, not be $\epsilon$-free.

Neumann and Seidl present guarded transformation in the context of \hes over distributive lattices with monotone operators
\cite{conf/csl/SeidlN99}. We stay with to the stipulations from Section~\ref{sec:prel} and only consider the powerset lattice with operators
$\Mudiam{a}$ and $\Mubox{a}$.

Neumann and Seidl give a guarded transformation procedure for a class of equation systems that contains the class of \hes obtained from
\mucalc-formulas in plain form. This transformation procedure runs in polynomial time and only produces a polynomial blowup. However, the
resulting equation system does not correspond to a flat \mucalc-formula. In Thorem~\ref{unravel-hes}, we see that turning an \hes into an
\mucalc-formula is likely to incur a blowup. Hence, the guarded transformation by Neumann and Seidl does not constitute a polynomial guarded
transformation for \mucalc-formulas.

The following example illustrates the loss in structure of the equation system representing an \mucalc-formula. Consider the following family of formulas:
\[
\mu X_1. \dotsc \mu X_n. X_1 \vee \dotsb \vee X_n \vee \Mudiam{a} \Big(\bigvee_{j = 1}^{m} \mu Y_j. \Mudiam{a} \big( Y_j \vee \bigvee_{i = 1}^n X_i\big)\Big)
\]
The associated equation system has a single block qualified with $\mu$ and looks like this:
\[
\begin{aligned}
X_1     &= X_2       &   Y_1 &=  \Mudiam{a} (Y_1 \vee X_1 \vee \dotsb \vee X_n) \\
        &\vdots      &       &\vdots  &       \\
X_{n-1} &= X_n       &   Y_m &=  \Mudiam{a} (Y_m \vee X_1 \vee \dotsb \vee X_n) \\ 
X_n     &= X_1 \vee \dotsb \vee X_n \vee \Mudiam{a}(Y_1 \vee \dotsb \vee Y_m)  & &   \\
\end{aligned}
\]
After the guarded transformation, the \hes looks like this:
\[
\begin{aligned}
X_1 &= \Mudiam{a} (Y_1 \vee \dotsb \vee Y_m)      & Y_1 &= \Mudiam{a} (Y_1 \vee X_1 \vee \dotsb \vee X_n)\\
    &\vdots                                     &     &\vdots \\
X_n &= \Mudiam{a} (Y_1 \vee \dotsb \vee Y_m)      &Y_m  &= \Mudiam{a} (Y_m \vee X_1 \vee \dotsb \vee X_n)\\
\end{aligned}
\]
This \hes has a variable dependency graph that is not a tree with back edges.

Note that the transformation of Neumann and Seidl produces downwards guarded formulas.

\subsection{Unraveling of Vectorial Formulas and \hes} 
We investigate how the different syntactic variants of \mucalc can be converted into each other.
\begin{lemma}
\label{hes-to-vec}
Any \hes can be converted into an equivalent vectorial formula with only polynomial blowup. This translations keeps guardedness, but the resulting vectorial formula will not be downwards guarded nor $\epsilon$-free.
\end{lemma}
\begin{proof}
The desired partition for the equations is already present from the block structure of the \hes. It remains to modify the equations such that no variable dependency edges go from a block into a block that is a successor block, but not a direct successor. So assume there is an equation $X = e_X$ in block $S_i$ that mentions a variable $Y$ in a block $S_j$ such that $j \geq i+2$. Introduce new variables $H_ {i+1},\dotsc,H_{j-1}$ with associated equations $H_{k} = H_{k+1}$ for all $k < j-1$ and $H_{j-1} = Y$. Moreover, replace all occurences of $Y$ in $S_i$ by $H_{i+1}$. It is not hard to see that the resulting equation system is equivalent. By repeating this procedure for all offending variables, the equation system can be made a vectorial formula. 

Since each variable in a block that is not the first induces at most $k$ intermediate variables, where $k+2$ is the number of blocks, the blowup is polynomial, namely at most quadratic in the number of variables. Moreover, no new unguarded cycles are introduced, but downwards guardedness is obviously lost.
\end{proof}

\setcounter{footnote}{2}

\begin{theorem}
\label{unravel-hes}
Any \hes with $n$ variables can be transformed into an equivalent flat \mucalc formula with blowup factor $2^{n-1}$.
\footnotetext{An earlier version \cite{arxiv-version} of this article contained an incorrect version of the transformation 
in Theorem~\ref{unravel-hes}.}
\end{theorem}
\begin{proof}
Let $\mathcal{S} = \{(S_1, \sigma_1), \dotsc, (S_k, \sigma_k)\}$ be an \hes. Let $\mathcal{X}$ denote the set of variables in $\mathcal{S}$. Let $<'_{p}$ denote a topological sorting of the priority order such that the entry variable is the maximal element,  and let $>'_{p}$ denote its converse. For all $X \in \mathcal{X}$, let $\sigma_{X} = \sigma_i$ if $X = e_X \in S_i$. 

Clearly, an equation of the form $X = e_X$ from a $\sigma$-block can be converted into an \mucalc-formula of the form $\sigma X.\psi_X$, with variables $Y$ occurring in $e_X$ being either free or another formula $\psi_Y$ being plugged in there. More precisely, for every $\sigma$-variable $X \in \mathcal{X}$ and every subset $\mathcal{Y} \subseteq \mathcal{X}$ with $X\in \mathcal{Y}$, we define a formula $\sigma X. \psi_X^{\mathcal{Y}}$ with $\psi_X^{\mathcal{Y}} = t(e_X)$ according to
\[
\begin{aligned}
t(q)                    &=  q \\
t(\overline{q})         &=  \overline{q} \\
t(\psi_1 \wedge \psi_2) &=  t(\psi_1) \wedge t(\psi_2) \\
t(\psi_1 \vee \psi_2)   &=  t(\psi_1) \vee t(\psi_2) \\
t(\Mudiam{a} \psi)      &=  \Mudiam{a} t(\psi) \\
t(\Mubox{a} \psi)       &=  \Mubox{a} t(\psi) \\
t(X)                    &=  X \\
t(X')                   &=  X'\text{ if }X' \in \mathcal{Y}\text{ and }X' >'_p X  \\
t(X')                   &=  \sigma_{X'} X' \psi_{X'}^{\mathcal{Y}\cup\{X'\}}\text{ if }
X' <'_p X \\
t(X')                   &=  \sigma_{X'} X' \psi_{X'}^{\mathcal{Y}\cup\{X'\}\setminus\{Z\colon Z <'_p X'\}}\text{ if }X'\notin\mathcal{Y}\text{ and  }X' >'_p X. \\
\end{aligned}
\]
If $Z$ is the entry variable of $\mathcal{S}$, the formula $\varphi_{\mathcal{S}} = \sigma_Z Z.\psi_Z^{\{Z\}}$ is equivalent to $\mathcal{S}$. Since the translation for the non-fixpoint operators is obviously correct, it is enough to show that modal operators are properly nested in order to show this. We observe that no formula of the form $\sigma_Z Z.\psi_Z^{\mathcal{Z}}$ has free variables $Z'$ such that $Z <'_p Z$. Moreover, the nesting of the formulas is finite. Since for each new fixpoint nesting, either a variable is added or a variable is added and all variables below it are cleared from the set $\mathcal{Z}$, the nesting process is finite: in order to remove a variable, a variable that is higher in the priority order has to be added. This can only happen a finite number of times; if $Z$ has $k$ variables above it, then $Z$ is removed no more than $\lfloor2^{k-2}\rfloor$ times. This leaves the maximal nesting depth at $2^{n-1}$, where $n$ is the number of equations. Moreover, since no formula of the form $\sigma_Z Z.\psi_Z^{\mathcal{Z}}$ has free variables $Z'$ such that $Z <'_p Z$, we can reuse formulas such that the formula DAG has exactly $2^{n-1}$ nodes when restricted to fixpoint formulas. Hence, the total size of the formula can be bounded from above by $2^{n-1}\cdot|e|$, where $e$ is the equation of maximal size.
\end{proof}

Eliminating $\epsilon$-transitions from an alternating parity tree automaton is treated in the literature in several places. Since an alternating parity tree automaton is just another way of presenting \hes, we briefly consider the problem, too. Wilke \cite{conf/fsttcs/Wilke99} gives an argument that elimination of $\epsilon$-procedures can be done with exponential blowup, but keeping a linear number of states. The latter actually has to be replaced with a quadratic blowup in the number of states. 

Vardi \cite{DBLP:conf/icalp/Vardi98} considers the problem in the more general framework of two-way automata. Finally, the guarded transformation procedure of Neumann and Seidl \cite{conf/csl/SeidlN99} for general \hes can be modified to yield a  procedure that eliminates $\epsilon$-transitions. Since Wilke's proof does not directly present an algorithm and Vardi's proof caters to a much more general framework, we give a variant of Neumann and Seidl's procedure.
\begin{lemma}
Let $\mathcal{S} = \{(S_1, \sigma_1),\dotsc,(S_k,\sigma_k)\}$ be an \hes or an alternating parity tree automaton with $n$ equations of total size $\Sigma_{i \leq n} e_i = m$. Then there is an $\epsilon$-free \hes $\mathcal{S}'$ with $k$ blocks, $nk$ equations and of exponential size.
\end{lemma}
\begin{proof}
For each variable $X$ in block $S_i$ introduce variables $X^j$ in block $S_j$ for all $j > i$. The new variables inherit the old transitions, i.e $e_{X^j} = e_{X}$. 

Note that the right-hand sides of the equations can be seen as elements of the free distributive lattice over the set of atoms $\mathcal{P} \cup \{\overline{p}\colon p \in \mathcal{P}\} \cup \mathcal{X} \cup \big( \bigcup_{a \in \Act} \{\Mudiam{a},\Mubox{a}\} \times \mathcal{X}\big)$. The height of this lattice, i.e. the maximal length of a strictly ascending or descending chain, can be bounded from above by $H = 2 |\mathcal{P}| + 3 | \mathcal{X} |$.

We can eliminate $\epsilon$-cycles from an equation $Z = e_Z$ from a block qualified with $\sigma$ in the following way: let $e^0_Z = e_Z
[\hat{\sigma}/Z]$, and let $e^{i+1}_Z = e_Z [e^i_Z /Z]$. The Knaster-Tarski Theorem yields the existence of a $j \leq H$ such that $e_Z
\equiv e^j_Z$. Moreover, $e^j_Z$ does not contain any $\epsilon$-transitions towards $Z$. In order to avoid doubly-exponential blowup, it is
convenient to convert all intermediate equations into disjunctive normal form. This increases the size of an equation exponentially and
takes exponential time. However, since the height of the lattice in question is $H$, no expression exceeds size $2^H$. Moreover, repeating
the conversion to normal form $H$ times still takes only time in $\mathcal{O}(H*2^H)$.

Let $X_1,\dotsc,X_{m'}$ be an enumeration of the variables compatible with the priority order, and assume that $\epsilon$-loops have been removed as per above. Eliminate $\epsilon$-transitions towards variables $X_i$ the following way, starting with the lowest variable $X_{m'}$ and ending with $X_1$: in an equation $X_j = e_{X_j}$, replace all unguarded $\epsilon$-transitions towards $X_i$ with $e'_{X_i}$, where $e'_{X_i}$ is obtained from $e_{X_i}$ by replacing all $Y$ from blocks with lower priority than $X_i$ with $Y^{i}$. After we have done this for $X_i$, the system does not contain any $\epsilon$-transitions towards $X_i$, so after the procedure finishes, the system is $\epsilon$-free. By a normal form argument as above, the overall blowup in the system does not exceed one exponential.

In order to argue why such a replacement preserves the semantics, consider the framework of a parity automaton. Instead of doing an
$\epsilon$-transition towards a state and then doing more transitions towards another state, we do the transitions to the third state right
away. This is correct as long as we record the parity of the state we skipped\footnote{This is the problem with Wilke's construction.}. The
additional copies of the states with low priority server this purpose. Obviously, we can skip a priority if we transition towards a higher
priority later. Correctness of the process follows by induction.
\end{proof}


%% file: lower.tex
\section{Lower Bounds}
\label{sec:lower}

In this section we show that guarded transformation for vectorial formulas and \hes is at least as hard---modulo polynomials---as parity game solving. The problem of whether or not the latter is possible in polynomial time has been open for a long while. We also show that unfolding a formula in vectorial form or an \hes into an equivalent non-vectorial formula is at least as hard as parity game solving. The core of the proofs is a product construction similar to that in Kupferman et al. \cite{KVW00}.

\begin{theorem}[Product Construction]
\label{thm:reductionboolcalc}
For every LTS $\Transsys$ 
 and every closed $\varphi \in \mucalc$ there is an LTS $\Transsys'$ with a single state $v_0$ and such that for every state $s_0$ in $\Transsys$, there is a vectorial  $\varphi'_{s_0} \in \boolcalc$ such that 
\begin{enumerate}
\item \label{reductionboolcalc-item1} $\Transsys, s_0 \models \varphi$ iff $\Transsys', v_0 \models \varphi'_{s_0}$,
\item \label{reductionboolcalc-item2} $|\Transsys'| = \mathcal{O}(|\varphi| \cdot |\Transsys|)$, and
\item \label{reductionboolcalc-item3} $|\varphi'_{s_0}| = \mathcal{O}(|\varphi|\cdot |\Transsys|)^2$.
\end{enumerate}
\end{theorem}
\begin{proof}
Let $\Transsys = (S,\Transition{}{}{},\ell)$ and let $\varphi$ be defined over propositions $\Prop$ and variables $\Var$. W.l.o.g.\ we assume that $S = \{1,\dotsc,m\}$ for some $m \in \Nat$. Define new sets of propositions $\Prop' := \Prop \times S$ and variables $\Var' := \Var \times S$. We write $X_s$ and $q_s$ instead of $(X,s)$ and $(q,s)$. 

Let $\Transsys' = (\{v_0\},\emptyset,\ell')$ consist of a single state with the following labeling: $q_s \in \ell'(v_0)$ iff $q \in \ell(s)$.

Next we give an inductively defined transformation $\operatorname{tr} \colon S \times \mucalc \to \boolcalc$   which turns an \mucalc 
formula over $\Var$ and $\Prop$ into a vectorial \boolcalc formula over $\Var'$ and $\Prop'$.  
\begin{align*}
\operatorname{tr}_s(q) \enspace &= \enspace q_s \\
\operatorname{tr}_s(\overline{q}) \enspace &= \enspace \overline{q_s} \\
\operatorname{tr}_s(X) \enspace &= \enspace X_s \\
\operatorname{tr}_s(\psi_1 \vee \psi_2) \enspace &= \enspace \operatorname{tr}_s(\psi_1) \vee \operatorname{tr}_s(\psi_2) \\
\operatorname{tr}_s(\psi_1 \wedge \psi_2) \enspace &= \enspace \operatorname{tr}_s(\psi_1) \wedge \operatorname{tr}_s(\psi_2) \\
\operatorname{tr}_s(\Mudiam{a}\psi) \enspace &= \enspace \bigvee \{ \operatorname{tr}_t(\psi) \mid t \in S \text{ with } \Transition{s}{a}{t} \} \\
\operatorname{tr}_s(\Mubox{a}\psi) \enspace &= \enspace \bigwedge \{ \operatorname{tr}_t(\psi) \mid t \in S \text{ with } \Transition{s}{a}{t} \} \\
\operatorname{tr}_s (\sigma X.\psi) \enspace &= 
\enspace \sigma X_s.
  \left\{\begin{array}{lcl}
                           X_{1}&.&\operatorname{tr}_{1}(\psi) \\
                           &\vdots \\
                           X_m&.&\operatorname{tr}_m(\psi) \\
  \end{array}\right\}
\end{align*}
Set $\varphi'_{s_0}$ to be $\operatorname{tr}_{s_0}(\varphi)$. It should be clear that $\varphi'_s$ is indeed a formula of \boolcalc.
For item~\ref{reductionboolcalc-item1} of the theorem, consider Stirling's local model-checking game for \mucalc \cite{stirling-game}.
It is not hard to see that $(s,\psi) \mapsto (v_0, \psi_s)$ maps positions in the model-checking game for $\Transsys,s_0$ and $\varphi$
isomorphically to positions in the game for $\Transsys',v_0$ and $\varphi'_{s_0}$. Moreover, strategy decisions map in the same manner,
hence Verifier has a winning strategy in one game if and only she has one in the other game, and the claim in
item~\ref{reductionboolcalc-item1} follows from that.

The size argument in item~\ref{reductionboolcalc-item2} follows because $\Transsys'$ has only one state, no transitions and $|\Prop|\times
|S|$ many propositions. It remains to argue for item~\ref{reductionboolcalc-item3}, regarding the size of $\varphi'$. Clearly, all steps except
the modal operators and the fixpoints do not produce blowup beyond $(|\varphi| \cdot |\mathcal{T}|)$ many subformulas. In the case of the
modal operators, each instance of a box or a diamond produces exactly one subformula for each edge in $\Transsys$, and edges with the same
target produce the same subformula. A subformula of the form $\sigma X. \psi$ will produce a vectorial fixpoint expression for each state in
$\Transsys$, each system of size $|S|$. However, all these systems are isomorphic except in their entry variable. By our stipulation from
Section~\ref{sec:prel}, these are considered the same subformula. Moreover, since each of the $X_s$ is available in every subformula of the
system, further nested systems are not nested with exponential blowup, but in a linear fashion. The claim for the formula size follows.
\end{proof}

\begin{theorem}
\label{gt-lower}
Guarded transformation for vectorial \mucalc-formulas or vectorial \boolcalc-formulas  is at least as difficult as solving parity games.  
\end{theorem}
\begin{proof}
We show that any polynomial guarded transformation for \boolcalc-formulas in vectorial form yields a  polynomial solution algorithm for parity games. The statement for \mucalc-formulas follows from that because   every guarded transformation for \mucalc-formulas is also one for \boolcalc-formulas.

Assume that there is a polynomial guarded transformation procedure $\tau$ for vectorial \boolcalc-formulas. Given a parity game, we can
treat it as an LTS with one accessibility relation and labellings for ownership and priority of states. Solving the parity game means
deciding whether the first player wins from the initial vertex, and this is equivalent to model-checking Walukiewicz' formula
\cite{MR1902097} for the corresponding priority. This formula is of size linear in the number of priorities, hence it is polynomial in the
size of the parity game. Via the product construction from Theorem~\ref{thm:reductionboolcalc}, we obtain a vectorial \boolcalc-formula
$\varphi$ that is also of size polynomial in the size of the parity game, and a one-state transition system $\Transsys$ such that $\Transsys
\models \varphi$ if and only if the first player wins the parity game. Consider $\tau(\varphi)$. Because $\tau$ runs in polynomial time, the
size of $\tau(\varphi)$ is still polynomial in the size of the parity game. Moreover, since the truth value of the \boolcalc-formula
$\varphi$ only depends on the state $v_0$, this must be true for $\tau(\varphi)$ as well. In effect, all modal operators introduced by
$\tau$ are vacuous and can be replaced by $\true$ in case of boxes, and by $\false$ in case of diamonds. The resulting vectorial formula is
again strictly boolean, but also guarded. Hence, it cannot contain any occurrences of fixpoint variables at all, since any such occurrence
would be unguarded. Therefore, all fixpoint quantifiers can be removed. The resulting formula is purely propositional and can be solved in
polynomial time by a simple bottom-up algorithm. This yields the desired polynomial solution for the parity game.
\end{proof}
The \boolcalc part of this theorem corresponds to Neumann and Seid's observation that ``finding equivalent guarded systems in general cannot be easier than computing solutions of hierarchical systems of Boolean equations'' \cite{conf/csl/SeidlN99}.

\begin{theorem}
\label{expand-lower}
Transforming a vectorial \mucalc-formula  or an \hes to a non-vectorial \mucalc-formula is at least as difficult as solving parity games. This holds even if the transformation only accepts $\epsilon$-free formulas and equation systems and even if it is allowed to produce unguarded formulas.
\end{theorem}
\begin{proof}
Assume that we have a polynomial transformation $\tau$ from $\epsilon$-free equation systems to \mucalc-formulas, and assume that we are given a parity game and a state in that parity game. We want to use the product construction from Theorem~\ref{thm:reductionboolcalc}, i.e. we want to construct a vectorial formula $\varphi'$ and a one-state transition system $\Transsys',v_0$ such that $\Transsys',v_0 \models \varphi'$ if and only if Verifier wins the parity game from the given state. 

Unfortunately, the vectorial formula $\varphi'$ from this theorem does not contain any modal operators, so it is far from being guarded, let
alone $\epsilon$-free. To remedy this problem, replace the structure $\Transsys',v_0$ by a version $\Transsys'',v_0$ that has a loop at the
only vertex. Clearly, a formula or equation system without modal operators is true in the old structure if and only if it is true in the new
structure. Moreover, for any $\mucalc$-formula or equation $\psi$, we have $\Transsys'',v_0 \models \psi$ if and only if $\Transsys'',v_0
\models \Diamond\psi$. Hence, we can replace every occurrence of a fixpoint variable $X$ or a fixpoint subformula $\sigma X. \psi_X$ in
$\varphi'$ by $\Diamond X$ and $\Diamond \sigma X.\psi_X$ without changing the truth value of the formula. The resulting formula $\varphi''$
is now $\epsilon$-free, and we can apply $\tau$. Since $\tau$ produces equivalent formulas, we have
\[
\Transsys'',v_0 \models \tau(\varphi'') \enspace \Leftrightarrow \enspace \Transsys'',v_0 \models \varphi'' 
\enspace \Leftrightarrow \enspace \Transsys',v_0 \models\varphi' \ .
\]
But since $\Transsys'',v_0 \models \psi$ if and only if $\Transsys'',v_0 \models \Diamond\psi$ for all $\psi$, and similar for modal boxes, we can replace any occurrence of modal operators of the form $\Diamond \psi$ or $\Box \psi$ by $\psi$. The resulting formula $\varphi'''$ does not contain any modal operators, and $\Transsys'',v_0 \models \varphi'''$ if and only if the first player wins the parity game from the given state. Since \mucalc-formulas without modal operators can be model-checked in polynomial time \cite{conf/csl/SeidlN99}, and since all the transformations above are polynomial, the transformation $\tau$ gives rise to a polynomial procedure for solving parity games.
\end{proof}


%% file: eval.tex
\section{Conclusion}
\label{sec:eval}

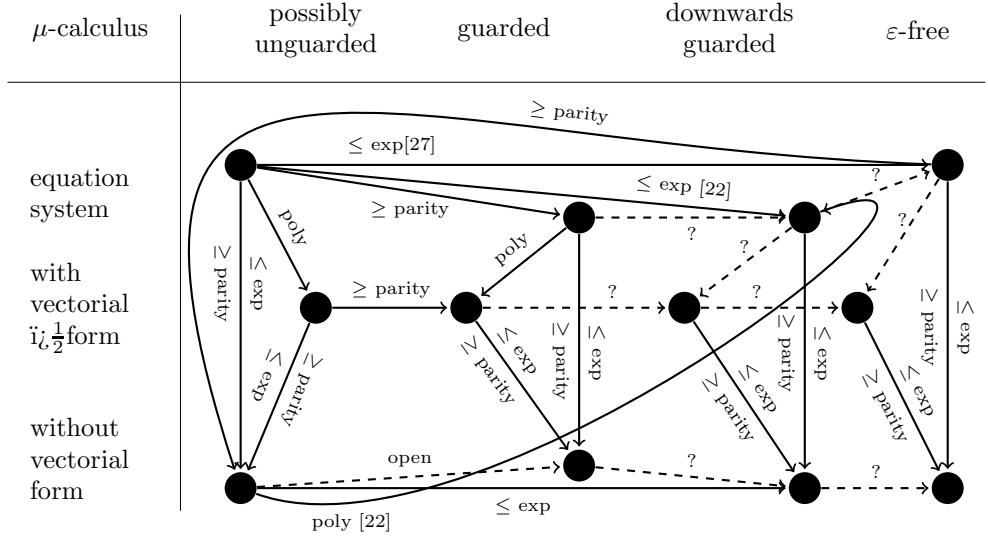
\begin{figure}[t]
\begin{tikzpicture}[allow upside down]
  \tikzstyle{blop}=[fill=black,shape=circle,inner sep=1.5mm]

  \node at (2,6.7) {$\mu$-calculus}; 
  \node at (5,6.7) {\parbox{2cm}{\centering possibly \\ unguarded}}; 
  \node at (7.5,6.7) {guarded};
  \node at (10.5,6.7) {\parbox{2cm}{\centering downwards \\ guarded}};
  \node at (13,6.7) {$\epsilon$-free};
  \node at (2.2,4.5) {\parbox{2cm}{equation \\ system}}; 
  \node[blop] (lu)  at (4,4.9) {}; 
  \node[blop] (mlu) at (8.5,4.2) {};
  \node[blop] (mru) at (11.5,4.2) {};
  \node[blop] (ru)  at (13.4,4.9) {}; 
  \node at (2.2,3) {\parbox{2cm}{with \\ vectorial \\�form}}; 
  \node[blop] (lm)  at (5,3.0) {}; 
  \node[blop] (mlm) at (7,3.0) {};
  \node[blop] (mrm) at (9.9,3.0) {};
  \node[blop] (rm)  at (12.2,3.0) {}; 
  \node at (2.2,1) {\parbox{2cm}{without \\ vectorial \\ form}}; 
  \node[blop] (ld)  at (4,.6) {}; 
  \node[blop] (mld) at (8.5,.9) {}; 
  \node[blop] (mrd) at (11.5,.6) {}; 
  \node[blop] (rd)  at (13.4,.6) {}; 

  \draw[very thin] (3.2,0.3) -- (3.2,6.9);
  \draw[very thin] (0.9,6.0) -- (13.9,6.0);

  \draw[thick,->,dashed,shorten >=1pt] (ld) -- node [above] {\scriptsize {open}} (mld) ;
   \draw[thick,->,shorten >=1pt] (ld) -- node [below] {\scriptsize $\leq$ exp } (mrd) ;
   \draw[thick,->,shorten >=1pt] (ld)  to[out=-20,in=20] node [pos=0.12,below] {{\scriptsize poly \cite{conf/csl/SeidlN99}}}(mru) ; 

    \draw[thick,->,shorten >=1pt] (lu) -- node [below] {{\scriptsize {$\ge$ parity}}} (mlu);
  \draw[thick,->,shorten >=1pt] (lm) -- node [above,sloped] {{\scriptsize {$\ge$ parity }}}   (mlm);
   \draw[thick,->,shorten >=1pt] (lu) -- node [above,sloped] {{\scriptsize {poly}}}  (lm);
   \draw[thick,->,shorten >=1pt] (mlu) -- node [rotate=180,above,sloped] {{\scriptsize {poly}}}  (mlm);
   \draw[thick,->,shorten >=1pt] (lu) -- node [pos=0.35,above,sloped] {{\scriptsize {$\le$ exp }}} node [pos=0.35,below,sloped] {\scriptsize {$\ge$ parity} } (ld);
   \draw[thick,->,shorten >=1pt] (lm) -- node [pos=0.4,above,sloped] {{\scriptsize {$\ge$ parity}}} node [pos=0.4,below,sloped] {{\scriptsize {$\le$ exp}}}  (ld);
   \draw[thick,->,shorten >=1pt] (mlm) -- node [pos=0.3,below,sloped] {{\scriptsize {$\ge$ parity}}} node [pos=0.3,above,sloped] {{\scriptsize {$\le$ exp}}} (mld);
   \draw[thick,->,shorten >=1pt] (mru) -- node [above,sloped] {{\scriptsize {$\le$ exp}}} node [below,sloped] {\scriptsize $\ge$ parity }  (mrd);

\draw[thick,->,shorten >=1pt] (lu) -- node [pos=0.8,above,sloped] {{\scriptsize {$\le$ exp \cite{conf/csl/SeidlN99} }}} (mru);

  \draw[thick,->,dashed,shorten >=1pt] (mld) -- node [above] {\scriptsize {?}} (mrd) ;
  \draw[thick,->,dashed,shorten >=1pt] (mrd) -- node [above] {\scriptsize {?}} (rd) ;
  \draw[thick,->,dashed,shorten >=1pt] (mlm) -- node [pos=0.7,above] {\scriptsize {?}} (mrm) ;
  \draw[thick,->,dashed,shorten >=1pt] (mrm) -- node [pos=0.4,above] {\scriptsize {?}} (rm) ;
  \draw[thick,->,dashed,shorten >=1pt] (mlu) -- node [below] {\scriptsize {?}} (mru) ;
  \draw[thick,->,dashed,shorten >=1pt] (mru) -- node [above] {\scriptsize {?}} (ru) ;
  \draw[thick,->,shorten >=1pt] (ru) -- node [above,sloped] {\scriptsize {$\le$ exp}} node [below,sloped] {\scriptsize $\ge$ parity }  (rd) ; 
 
   \draw[thick,->,shorten >=1pt] (mlu) -- node [pos =0.55,above,sloped] {\scriptsize {$\le$ exp}} node [pos=0.59,below,sloped] {\scriptsize $\ge$ parity }  (mld) ; 
   
  \draw[thick,->,dashed,shorten >=1pt] (mru) -- node [above] {\scriptsize {?}} (mrm) ; 
  \draw[thick,->,shorten >=1pt] (mrm) -- node [above,sloped] {\scriptsize {$\le$ exp}} node [below,sloped] {\scriptsize $\ge$ parity }  (mrd) ; 
 \draw[thick,->,dashed,shorten >=1pt] (ru) -- node [above] {\scriptsize {?}} (rm) ; 
  \draw[thick,->,shorten >=1pt] (rm) --  node [above,sloped] {\scriptsize {$\le$ exp}} node [below,sloped] {\scriptsize $\ge$ parity }  (rd) ; 
  \draw[thick,->,shorten >=1pt] (lu) -- node [pos=0.2,above] {\scriptsize {$\le$ exp}\cite{conf/fsttcs/Wilke99}} (ru) ; 

   \draw[thick,->,shorten >=1pt] (ru)  to[out=178,in=108,distance=7.3cm] node [pos=0.25,above,sloped,rotate=180] {{\scriptsize $\ge$ parity }}(ld) ;

\end{tikzpicture}
\vskip-3mm
\caption{State of the art on guarded transformations for the modal $\mu$-calculus.}
\label{fig:gt}
\end{figure}

In Section~\ref{sec:oldgt} we showed that the known guarded transformations produce an exponential blowup in the worst case. We also presented
exponential algorithms for the elimination of $\epsilon$-transitions in alternating parity tree automata and the translation from \hes or
vectorial formulas into flat \mucalc-formulas.

In Section~\ref{sec:lower} we showed that a polynomial guarded transformation for vectorial formulas entails a polynomial solution algorithm
for parity games, and the same is true for \hes. We also proved that polynomial translation from vectorial formulas to non-vectorial
formulas, or from \hes to \mucalc-formulas, yields the same. The existence of a polynomial guarded transformation for \emph{non-vectorial}
formulas is still open, and it is possible that such a transformation exists without yielding a polynomial solution algorithm for parity
games. Figure~\ref{fig:gt} gives an overview over the current state of research.

\paragraph*{Consequences and Corrections.}
Several constructions and procedures that deal with the modal $\mu$-calculus directly or indirectly
use guarded transformation. Often they use the (possibly) false claim that guarded transformation can
be done at a linear or quadratic blow-up. We examine consequences of the fact that according to current
knowledge, guarded transformation is exponential.

As a first step, it is interesting to check whether any of the results from
Kupferman/Vardi/Wolper's seminal paper on automata for branching-time temporal logics \cite{KVW00} crucially 
rely on a polynomial guarded transformation. Fortunately, this is not the case. The product automaton 
construction \cite[Prop.~3.2]{KVW00} also works if the input automaton is
not $\epsilon$-free, and the resulting product automaton has no $\epsilon$-transitions. This allows the 
subsequent Thm.~3.1 to be applied, which needs $\epsilon$-free automata. Hence, all results of
\cite{KVW00} that rely on a polynomial guarded transformation remain true.

In \cite{journals/tcs/HenzingerKM06}, the reliance on a polynomial guarded transformation causes
problems. It is not immediately obvious that the complexity results for satisfiability of existential
and universal \mucalc and alternation-free \mucalc (claimed to be NP-complete) as well as derived
results should hold for unguarded formulas. A preceding transformation into guarded form will exhaust
the complexity limitations, so further investigation on this work is necessary.

Mateescu claims that model checking for $\mucalc$ on acyclic structures can be done in polynomial time. He
observes that, on acyclic structures, least and greatest fixpoints coincide. Thus, the alternation hierarchy
collapses to its alternation-free fragment on such structures. It is known that this fragment can be model-checked
in linear time \cite{journals/fmsd/CleavelandS93}. However, least and greatest fixpoints only coincide for
guarded formulas: clearly $\mu X.X$ and $\nu X.X$ are not equivalent, not even on acyclic structures. The 
collapse result is still true but with current technology at hand, we need to assume an exponential
blow-up in formula size. Thus, model checking guarded formulas on acyclic structures can be done in
polynomial time, arbitrary formulas still require exponential time. 

\paragraph*{Automata, $\epsilon$-transitions and Vectorial Form.}
It is standard practice to construct alternating parity tree automata from guarded \mucalc-formulas, or weak alternating tree automata from guarded alternation-free formulas \cite{focs91*368,Wilke:2001:BBMS,GKL:Gandalf12}. The resulting automata are of size linear in the size of 
the formula. The situation is different for unguarded (alternation-free) \mucalc-formulas because of the absence of a polynomial guarded transformation. Currently, we need to assume a blow-up that is exponential in the 
size of the formula when translating arbitrary, and therefore possibly unguarded, \mucalc-formulas into alternating parity tree automata.
If $\epsilon$-transitions are allowed in alternating parity tree automata then it is possible to translate arbitrary $\mucalc$-formulas into such automata at a linear blow-up only; the known constructions can be modified accordingly. Since, by Theorem~\ref{gt-lower}, eliminating $\epsilon$-transitions from alternating automata must be considered exponential, obtaining an $\epsilon$-free alternating automaton from an unguarded \mucalc-formula incurs an exponential blowup by current state of research.

The above means that the size of an alternating automaton cannot be measured in terms of number of states.
Instead, such a notion of size has to include the size of the transition relation, which can easily be exponentially larger. An example of the confusion that the wrong measure can cause is Kupferman/Vardi's work on alternating automata. They show that nonemptiness of weak alternating automata can be checked in linear time when size is measured including the transition relation \cite{KVW00}. This result is then used in a context where the size is measured in the number of states \cite{conf/stoc/KupfermanV98}. We do believe that this claim, and subsequent results, are to be questioned.

Another problematic custom is that authors often use vectorial \mucalc and flat \mucalc interchangebly, or sometimes even alternating parity automata and \mucalc. Since translating parity tree automata back into \mucalc-formulas is exponential as well, the latter seems inappropriate. An example of the former is \cite{DBLP:journals/tocl/KupfermanV05}, where weak alternating automata are translated linearly into vectorial alternation free \mucalc, referred to as \mucalc only. This seems confusing, since unraveling the vectorial formula obtained will incur exponential blowup. We think that vectorial form has its use, in particular because weak automata translate so easily into vectorial form, but the use of vectorial form should always be clearly labeled as such, in order to avoid hidden exponential gaps.

\paragraph*{Further Research.}
The existence of a polynomial guarded transformation for flat \mucalc is still open, and it is also not known whether the existence of such a procedure would entail parity games being solvable in polynomial time. Interestingly enough, all known guarded transformation procedures produce downward guarded formulas, and we strongly suspect that any reasonable candidate for a polynomial guarded transformation will do so as well. This is because the notion of guardedness as such seems too weak: aiming for downward guardedness gives enough structure to the formula to establish reasonable induction invariants. On the other hand, showing that plain guarded transformation is subject to some lower bound suffers from the unstructured notion of guardedness. This is taken to an extreme in Berwanger's two-variable version \cite{DBLP:journals/sLogica/Berwanger03} of the Walukiewicz formulas: his formulas are not guarded at all, but are \emph{de facto} guarded. The means that, in a tableau, any two iterations of the same fixpoint variable will have a modal operator in between. Judging from the overall picture, we think that downward guardedness might be the right target to attack.

There are several arrows in Figure~\ref{fig:gt} where a transformation can be done with exponential blowup, but the lower bound is only that
any polynomial procedure would also entail a polynomial solution algorithm for parity games. Closing theses gaps can be done in two ways:
giving a polynomial solution algorithm for the problem---and establishing that parity games can be solved in polynomial time in the
process---or establishing an exponential lower bound for the transformation. 

Finally, the converse questions are also open: does a polynomial solution algorithm for parity games entail a polynomial guarded transformation procedure? Does it entail a polynomial procedure to unravel vectorial formulas or \hes into flat \mucalc-formulas? Of course, the ability to solve parity games in polynomial time allows polynomial model-checking for vectorial formulas, but the problems of guarded transformation, respectively of unraveling formulas, seem to be independent of that.
